\documentclass{article}
\usepackage{arxiv}

\final

\usepackage{tikz}

\newtheorem{theorem}{Theorem}
\newtheorem{lemma}[theorem]{Lemma}
 
\theoremstyle{definition}

\definecolor{ForestGreen}{rgb}{.13,.54,.13}

\newcommand{\cakeprt}{\mathbf{x}}
\newcommand{\choreprt}{\mathbf{\widehat{x}}}
\newcommand{\simplex}{\Delta^{m-1}}
\newcommand{\smallsimplex}{\delta^{m-1}}

\begin{document}

\title{Cutting a Cake Fairly for Groups Revisited}
\markright{Cutting a Cake Fairly}
\author{Erel Segal-Halevi and Warut Suksompong}

\maketitle

\begin{abstract}
Cake cutting is a classic fair division problem, with the cake serving as a metaphor for a heterogeneous divisible resource.
Recently, it was shown that for any number of players with arbitrary preferences over a cake, it is possible to partition the players into groups of any desired size and divide the cake among the groups so that each group receives a single contiguous piece and every player is envy-free.
For two groups, we characterize the group sizes for which such an assignment can be computed by a finite algorithm, showing that the task is possible exactly when one of the groups is a singleton.
We also establish an analogous existence result for chore division, and show that the result does not hold for a mixed cake.
\end{abstract}

\section{Introduction.}

After a strong annual performance, the owner of a company decides to reward her 20 employees by renting them a holiday cottage for the last two weeks of the year.
Due to capacity constraints, she needs to split the employees into two groups of ten, divide the two-week period into two contiguous intervals, and assign one interval to each group.
The employees have different preferences on the time they would like to spend at the cottage.
Can the owner make sure that the resulting assignment is \emph{envy-free}, that is, no employee believes that the time slot given to the other group is better than the one given to his or her own group?

In a recent \textsc{Monthly} note, Segal-Halevi and Suksompong~\cite{SegalhaleviSu21} showed that for any preferences that the employees may have, there always exists an assignment fulfilling the owner's wishes.
In fact, the result holds even if there are $n$ employees and the owner would like them to vacation at the cottage in groups of size $k_1,k_2,\dots,k_m$ in this order, for any $n$, $m$, and any $k_j$'s that sum to $n$.
This result generalizes an influential theorem by Stromquist~\cite{Stromquist80}, Su~\cite{Su99}, and Woodall~\cite{Woodall80} (the first two also in this \textsc{Monthly}), which concerns the case where $k_1=\dots=k_m=1$, i.e., the employees go to the cottage one at a time.
Yet, two important questions remain:
\begin{itemize}
\item \emph{How} can the owner compute a desired assignment? Is it even possible to do so?
\item What if instead of dividing a \emph{good} like holiday accommodation, we were to divide a \emph{chore} such as housework, which yields a disutility to everyone involved? Or a mixture of goods and chores?
\end{itemize}
In this article, we provide comprehensive answers to both of these questions.

\section{Model.}
We switch to the language of \emph{cake cutting}---the classical framework for studying how to divide a resource fairly \cite{BramsTa96,RobertsonWe98}.
In cake cutting, we have a set of $n$ players along with a cake in the form of an interval $[0,1]$.
A \emph{partition} of the cake 
is specified by nonnegative real numbers $x_1,\dots,x_m$ such that $\sum_{j=1}^m x_j = 1$, where $x_j$ denotes the length of the $j$th piece.
For each partition $(x_1,\dots,x_m)$, a player \emph{prefers} one or more pieces.\footnote{
Formally, each player's preference is represented by a function from the standard  simplex $\simplex$ to a nonempty subset of $\{1,2,\dots,m\}$.
Note that this model allows a player to prefer, for example, the piece of length $0.8$ in the partition $(0,0.2,0.8)$
and the piece of length $0.2$ in the partition $(0.2,0,0.8)$, even though the partitions are physically identical. We could add a requirement that each player always prefers the same physical piece(s) in such equivalent partitions; this would not affect our results except in Section \ref{sec:mixed}.
}
The result by Segal-Halevi and Suksompong~\cite{SegalhaleviSu21} is based on reducing the problem to the case of singleton groups and applying the theorem of Stromquist~\cite{Stromquist80} and Su~\cite{Su99}.
Both of these results work under the following two mild assumptions:
\begin{enumerate}
\item \emph{Hungry players.} Players never prefer an empty piece.
\item \emph{Closed preference sets.} Any piece that is preferred for a convergent sequence of partitions is also preferred at the limiting partition.
\end{enumerate}

An \emph{envy-free assignment}
is 
a partition of the cake into $m$ contiguous pieces, together with a division of the players into $m$ groups with group $j$ containing $k_j$ players, such that each player in group $j$ prefers the $j$th piece from the left in the partition.\footnote{Fair division among \emph{groups} of agents has received increasing attention, not only in cake cutting \cite{SegalhaleviNi19,SegalhaleviSu21} but also in the allocation of indivisible goods \cite{KyropoulouSuVo20,ManurangsiSu17,ManurangsiSu22,SegalhaleviSu19,Suksompong18,Suksompong18-dissertation} and rent division \cite{GhodsiLaMo18}.}

\section{Protocols.}

We focus first on the case of two groups.
For this case, a protocol for computing an envy-free assignment is given in \cite{SegalhaleviSu21}:
each player $i$ marks a point $x_i\in[0,1]$ such that the player prefers both $[0,x_i]$ and $[x_i,1]$. The protocol then cuts the cake at a point~$y$ between the $k_1$th and the $(k_1+1)$st marks from the left, and assigns the players with the $k_1$ leftmost marks to $[0,y]$ and the remaining players to $[y,1]$.
However, this protocol relies crucially on an additional assumption that the players' preferences are \emph{monotone}, meaning that if a player prefers a piece $[0,y]$, he or she must also prefer the piece $[0,z]$ for any $z > y$, and analogously for $[y,1]$.
While monotonicity is a reasonable assumption in certain situations, it is relatively strong in comparison to the two assumptions needed for the existence result.
Indeed, in our introductory example, some employees may not want to spend too long at the cottage, and forgoing part of the allotted time can be frowned upon by the owner or the other group.
What, then, can we do when the preferences are possibly non-monotone?

For brevity, let us say that a player \emph{prefers the left piece} (resp., \emph{right piece}) \emph{at point $x\in[0,1]$} if the player prefers the interval $[0,x]$ (resp., $[x,1]$).
The hungry players assumption implies that every player prefers only the right piece at $0$ and prefers only the left piece at $1$.
A first idea that comes to mind is to ask the players for their entire preferences so that we can determine an appropriate cut point $x$, which we know must exist.
However, this may lead to an infinite protocol.
Specifically, given any infinite sequence of numbers $0 < a_1 < a_2 < \dots < 1$, it is possible that a player prefers the right piece at any point in $[0,a_1]$, the left piece at any point in $[a_1,a_2]$, the right piece at any point in $[a_2,a_3]$, and so on---such a preference is consistent with the hungry players and closed preference sets assumptions, but cannot be described in finite time.
A protocol that attempts to learn this preference in its entirety will never terminate.

To see what can be done with finite protocols, let us recall the simplest cake-cutting protocol: cut-and-choose. This protocol begins by asking one of the players a \emph{query}, which can be phrased as ``mark a point $x$ such that you are indifferent between the cake to the left of $x$ and the cake to the right of $x$.''
We slightly modify this query to handle the case in which there are several such ``indifference points:''
\begin{quote}
\normalsize
\emph{First-indifference-point.} Given a player $i$, return the smallest point $y \in [0,1]$ such that $i$ prefers both pieces at $y$.
\end{quote}
\begin{lemma}
\label{lem:first-indifference-point}
For every hungry player $i$ with closed preference sets, 
$i$'s first indifference point exists, 
it equals the smallest point at which $i$ prefers the left piece, and it is strictly between $0$ and $1$.
\end{lemma}
\begin{proof}
Let $y$ be the smallest point such that $i$ prefers the left piece at $y$. Observe that (i) $y$ exists because $i$ prefers the left piece at~$1$, (ii) $y > 0$ because $i$ prefers only the right piece at $0$, and (iii) $y$ is well-defined due to the closed preference sets assumption.
Then, since $i$ prefers the right piece at any point $z < y$, the closedness assumption tells us that $i$ also prefers the right piece at $y$.
Hence $y$ is the first indifference point. The hungry players assumption implies that this point $y$ is strictly between $0$ and $1$.
\end{proof}

When the preferences are monotone, the set of indifference points
is a (possibly degenerate) closed interval, and the answer to a First-indifference-point query is the leftmost point of that interval.
However, when the preferences are not monotone, the set of indifference points can be any closed subset of $(0,1)$.
Indeed, given such a subset~$C$, suppose that its rightmost point is $r < 1$, which is well-defined since $C$ is closed.
A hungry player may prefer the left piece at any point in $C\cup[r,1]$ and the right piece at any point in $[0,r]$.
The set of indifference points is then precisely $C$.

We show that the First-indifference-point query is sufficient for handling the special case in which the first group is a singleton.
For any positive integer $t$, denote by $[t]$ the set $\{1,2,\dots,t\}$.

\begin{theorem}
\label{thm:cake-pos}
Let $n$ be any positive integer and $m=2$.
For $n$ hungry players with closed preference sets, there is a finite protocol using the First-indifference-point query that computes an envy-free assignment with the first group containing one player and the second group containing $n-1$ players.
\end{theorem}

\begin{proof}
Our protocol makes a First-indifference-point query for every player---assume that the answers are $x_1,\dots,x_n$.

Let $i'\in[n]$ be such that $x_{i'} = \min\{x_1,\dots,x_n\}$.
The protocol cuts the cake at point $x_{i'}$, and assigns player~$i'$ to the left piece and the remaining players to the right piece.
Player~$i'$ prefers the left piece by definition of $x_{i'}$.
For any other player~$i$, 
Lemma~\ref{lem:first-indifference-point} implies that 
$x_i$ is the smallest point at which $i$ prefers the left piece, so
$i$ prefers the right piece at any point $z_{i} \le x_{i}$, and therefore also prefers the right piece at $x_{i'}\le x_{i}$.
\end{proof}

At this point, one might be tempted to think that the simple protocol in Theorem~\ref{thm:cake-pos} can be generalized to two groups with any player distribution.
However, and perhaps surprisingly, we show that \emph{no} finite protocol can compute an envy-free assignment when both groups consist of at least two players.
This impossibility holds even if we allow the protocol to make more general types of queries:\footnote{Our queries are intended to mirror those in the ``Robertson--Webb model''---see the discussion at the end of this section.}
\begin{quote}
\normalsize
\begin{itemize}
\item \emph{Next-indifference-point.} Given a player and a point $x\in[0,1]$, return the smallest point $y\ge x$ such that the player prefers both pieces at $y$ (or state that no such $y$ exists).
Note that this is a generalization of the First-indifference-point query, which corresponds to the special case where $x=0$.
\item \emph{Previous-indifference-point.} Given a player and a point $x\in[0,1]$, return the largest point $y\le x$ such that the player prefers both pieces at $y$ (or state that no such $y$ exists).
\item \emph{Evaluate.} Given a player and a point $x\in[0,1]$, return whether the player prefers the left piece, the right piece, or both pieces at $x$.
\end{itemize}
\end{quote}

While it may seem restrictive to consider only indifference points, we remark here that the three types of queries are together quite powerful.
For example, 
suppose we are given a point $x$, and we want to determine the smallest point $y\ge x$ such that the player prefers the left piece at $y$.
We can first make an Evaluate query at $x$.
If the left piece is preferred, the answer is $x$.
Else, we make a Next-indifference-point query at~$x$.
By an argument similar to that in Lemma \ref{lem:first-indifference-point}, the answer $y$ to this query is also our desired answer.
Analogously, we can determine the smallest point $y\ge x$ such that the player prefers the \emph{right} piece at $y$, or the largest point $y\le x$ such that the player prefers the left (or right) piece at $y$.

\begin{theorem}
\label{thm:cake-neg}
Let $n$ be any positive integer, $m=2$, and $k_1,k_2$ be positive integers such that $k_1+k_2 = n$ and $\min\{k_1,k_2\} \ge 2$.
There does not exist a finite protocol using the Next-indifference-point, Previous-indifference-point, and Evaluate queries that, for any $n$ hungry players with closed preference sets, computes an envy-free assignment with group $j$ containing $k_j$ players.
\end{theorem}

\begin{proof}
Assume for the sake of contradiction that such a protocol exists.
We will show how an adversary can answer the protocol's queries in such a way that after any finite number of queries, for any assignment that the protocol may output, there exist players' preferences consistent with the answers for which the assignment is not envy-free.
This is sufficient to obtain the desired contradiction.

During the run of the protocol, the adversary maintains a set of ``known intervals''---the protocol knows the preferences of all players at every point in these intervals. All of these intervals are closed.
Initially, only the degenerate intervals $[0,0]$ and $[1,1]$ are known.
Whenever the protocol makes a query, the adversary modifies the known intervals in such a way that
(a) the query can be answered using the information encoded in the known intervals, (b) the known intervals do not cover the entire cake, and (c) if the cake is cut anywhere inside a known interval, then at least $n-1$ players want the same piece. We now explain exactly how the adversary modifies the known intervals.

Suppose that the protocol asks player $i$ a query at point $x$.
If $x$ belongs to a known interval, the adversary extends the interval on both sides (unless the interval contains $0$ or $1$, in which case the protocol extends it on one side) so that the gap to the next known interval is reduced by half.
For example, if at the beginning a query is made at point $0$, the adversary extends the known interval $[0,0]$ to $[0,0.5]$.
On the other hand, if $x$ is outside any known interval, the adversary adds a new known interval containing~$x$ such that the gap to the adjacent known interval on each side is reduced by half.
For example, if at the beginning a query is made at point $0.4$, the adversary adds a new known interval $[0.2,0.7]$.

Now, the adversary determines the players' preferences in the newly added (or extended) known intervals as follows. 
All players $i'\neq i$ prefer only the right piece at any point added to a known interval.
For player $i$, for each maximal new portion $[a,b]$ of a known interval, the adversary chooses $a < a' < b' < b$, and lets player $i$ prefer only the right piece everywhere in $[a,a')\cup(b',b]$, only the left piece everywhere in $(a',b')$, and both pieces at $a'$ and $b'$.
For example, if at the beginning the protocol asks player $i$ a query at point $0.4$, then all players except $i$ prefer only the right piece at any point in $[0.2,0.7]$, while player~$i$ prefers only the right piece at, e.g., any $y\in[0.2,0.4)\cup(0.5,0.7]$, only the left piece at any $y\in(0.4,0.5)$, and both pieces at $0.4$ and $0.5$.

The only exception to the rules in the previous paragraph is if the query is made at a point in the rightmost known interval---in this case, the roles of ``right'' and ``left'' are reversed.
So, if at the beginning a query is made at point $1$ for player~$i$, then all other players prefer only the \emph{left} piece at any $y\in [0.5,1]$, while player~$i$ prefers only the left piece at, e.g., any $y\in [0.5,0.7)\cup(0.8,1]$, only the right piece at any $y\in (0.7,0.8)$, and both pieces at $0.7$ and $0.8$.

Observe that, after the modifications, player $i$ has indifference points both to the left and to the right of $x$ (or on one side if $x=0$ or $x=1$), and the preference of $i$ at~$x$ is known, so the adversary has sufficient information to answer any query based on the newly specified preferences.

Since the protocol is finite, it must eventually return a partition of the cake with some cut point $y$.
If $y$ is not in a known interval, then the adversary voluntarily answers a query for some player $i$ at point $y$, using the same rules as above, so now $y$ is in a known interval.

The adversary's answers guarantee that, at any point in a known interval, including~$y$, at least $n-1$ players prefer only the left piece, or at least $n-1$ players prefer only the right piece.
Since $\min\{k_1,k_2\}\ge 2$, it follows that the assignment returned by the protocol cannot be envy-free.

Finally, 
to demonstrate that there are complete preferences that are consistent with its answers,
the adversary completes the preferences of all players in all unknown intervals by letting all players prefer only the right piece at every point in each interval.
The only exception is the unknown interval adjacent to the rightmost known interval---for this unknown interval, the adversary chooses a point $z$ in it and lets all players prefer only the right piece at points to the left of $z$, only the left piece at points to the right of $z$, and both pieces at $z$.
It is easy to verify that the resulting preferences satisfy both the hungry players assumption and the closed preference sets assumption.
\end{proof}

Theorems~\ref{thm:cake-pos} and \ref{thm:cake-neg}, combined with the observation that an analog of Theorem~\ref{thm:cake-pos} yields a finite protocol using the Previous-indifference-point query for the case $k_2 = 1$, provide a complete characterization of the group sizes for which a finite protocol exists in the setting with two groups.

What happens when we desire a division into three or more groups?
It turns out that even for three or more \emph{singleton} groups (that is, $m\geq 3$ and $k_j=1$ for all $j\in[m]$), no finite protocol exists.
This follows from an impossibility result by Stromquist~\cite{Stromquist08}
for the \emph{additive} model.
In this model, each player $i$ has a utility function $u_i$, which is a nonnegative nonatomic measure on the cake, such that $u_i([0,1]) > 0$.
In each cake partition, player $i$ prefers piece $[x,y]$ if and only if $u_i([x,y])$ is maximum among the pieces in the partition.
The assumption that the functions $u_i$ are nonnegative implies the hungry players assumption; the assumption that they are nonatomic implies the closed preference sets assumption. Hence, the additive model is a special case of the model that we study.

In the additive model, a protocol is allowed to make two types of queries (which together form the \emph{Robertson--Webb model}~\cite{RobertsonWe98}): an \emph{Evaluate} query, returning a player's utility for a given piece of cake, and a \emph{Cut} query, returning a piece of cake for which a player has a certain utility starting from a given point.
Like our proof, Stromquist's proof uses an adversary argument where the adversary answers the protocol's queries in such a way that the protocol cannot produce an envy-free assignment with certainty after any finite number of queries.
Under our general preferences without cardinal utilities, it is unclear which queries should be allowed when there are three or more groups; moreover, in light of the impossibility result in the more specific model, it appears unlikely that a finite algorithm would exist for any reasonable set of permitted queries.
We will therefore go no further in this direction.\footnote{
As we noted above, Stromquist's impossibility result holds even when all players have preferences based on additive utility functions.
The same is true for our Theorem \ref{thm:cake-neg}: the preferences used in the proof can be represented by additive utility functions (however, unlike Stromquist's result, our result requires non-monotone preferences). 
To see this, for each player $i$, 
plot a continuous function between $(0,0)$ and $(1,1)$, crossing $1/2$ at every indifference point of $i$, so that the function is below $1/2$ in any interval in which $i$ prefers the right piece, and 
above $1/2$ in any interval in which $i$ prefers the left piece. Since every player has a finite number of indifference points in our construction, such a function $U_i: [0,1]\to [0,1]$ represents an additive utility function, which assigns to each interval $[x,y]$ the utility $U_i(y)-U_i(x)$.
}

The following table compares our results in this section with previous results.
Positive results are valid even for non-additive preferences, while negative results are valid even for additive preferences.

\begin{center}
\begin{tabular}{|c|c|c|c|}
 \hline
\textbf{ Monotonicity }& $m$ & $k_i$ & \textbf{Protocol?} \\
 \hline
 \hline
Monotone & $\geq 3$ & Arbitrary & No \cite{Stromquist08} \\
 \hline
Monotone & $2$ & Arbitrary & Yes \cite{SegalhaleviSu21} \\
 \hline
 Arbitrary & $2$ & $k_1=1$ or $k_2=1$ & Yes [Theorem~\ref{thm:cake-pos}] \\
 \hline
Non-monotone & $2$ & $k_1, k_2\geq 2$& No [Theorem~\ref{thm:cake-neg}] \\
 \hline
\end{tabular}
\end{center}

\section{Chore Division.}
Suppose now that, instead of a holiday cottage, there is a sensitive public building that must be secured during the night. 
We therefore wish to partition the night into $m$ shifts. In each shift $j\in[m]$, the building should be secured by $k_j$ guards.

When $k_j=1$ for all $j\in[m]$, the problem is an instance of the well-known \emph{chore division} problem, wherein an undesirable task is to be divided among players who prefer to receive as little of the task as possible according to their own preferences. 
The difference between chore division and cake cutting is expressed by replacing the hungry players assumption with the lazy players assumption:

\begin{enumerate}
\item$\!\!\!\!{^*}$  \emph{Lazy players}. 
Players never prefer a nonempty piece if an empty piece is available.
\end{enumerate}

\begin{lemma}
\label{lem:strong-lazy}
A lazy player with closed preference sets prefers all (and only) the empty pieces whenever there is at least one empty piece.
\end{lemma}
\begin{proof}
If exactly one piece is empty, then a lazy player by definition prefers only this piece. 
Consider now a partition $\cakeprt$ in which pieces $j_1,\dots,j_t$ are empty for some $t > 1$.
This partition $\cakeprt$  is the limit of $t$ sequences of partitions such that in each sequence $\ell\in [t]$, only piece $j_{\ell}$ is empty while the sizes of the pieces $j_{\ell'}$ for $\ell'\ne\ell$ converge to~$0$. 
A lazy player prefers piece $j_{\ell}$ in all partitions of sequence $\ell$;
therefore, a lazy player with closed preference sets prefers piece $j_\ell$ in the limit partition $\cakeprt$ too. This is true for every empty piece $j_\ell$.
\end{proof}

For individuals (i.e., singleton groups), the existence of an envy-free chore division was established in this \textsc{Monthly} by Su~\cite{Su99}. 
We show next that the general group result of Segal-Halevi and Suksompong~\cite{SegalhaleviSu21}, too, extends to chore division. 
The proof uses a reduction from group chore division to group cake cutting, which may be useful in other contexts beyond group division.\footnote{
Segal-Halevi and Suksompong~\cite{SegalhaleviSu21} reduced group cake cutting to individual cake cutting.
Why can we not use the same method for reducing group chore division to individual chore division?
The reason is that while their reduction preserves ``hungriness,'' it does not preserve ``laziness.''

More specifically, given a player in a group cake-cutting instance, the reduction constructs a player in an individual cake-cutting instance. To determine the preferences of the individual-instance player on an $n$-partition~$\cakeprt$, the reduction constructs an $m$-partition $\choreprt$ by uniting sequences of adjacent pieces: the union of pieces $1,\ldots k_1$ in $\cakeprt$ is the first piece in $\choreprt$; 
the union of pieces $k_1+1,\ldots k_2$ in $\cakeprt$ is the second piece in $\choreprt$; and so on.

Now, if some piece in $\cakeprt$ is \emph{nonempty}, then the corresponding piece in $\choreprt$ is nonempty too; therefore, hungriness of the group-instance player implies hungriness of the individual-instance player.
However, 
if some piece in $\cakeprt$ is \emph{empty}, then there is no reason to assume that the corresponding piece (or any other piece) in $\choreprt$ will be empty too; so laziness of the group-instance player does \emph{not} imply laziness of the individual-instance player.
}

\begin{theorem}
\label{thm:chores}
Let $n\geq m$ be any positive integers, and let $k_1,k_2,\dots,k_m$ be positive integers such that $\sum_{j=1}^m k_j = n$.
For any $n$ lazy players with closed preference sets, there is an $m$-partition of the chore and an envy-free assignment of the parts to $m$~groups, with group~$j$ containing $k_j$ players.
\end{theorem}

Before proving the theorem, we introduce some formal notation. 
We represent each $m$-partition of the cake or chore 
by a vector $\cakeprt \in \simplex$, where $\simplex$ is the standard simplex in $\mathbb{R}^m$ (i.e., for all $j\in[m]$, $x_j\geq 0$ and  $\sum_{j=1}^m x_j=1$). 
The value of $x_j$ indicates the length of piece $j$ (where the piece indices start from the left).

The preferences of a player are represented by a \emph{demand function} $h:\simplex\rightarrow 2^{[m]}$ that assigns to each partition a nonempty subset of preferred pieces.
Given $j\in[m]$ and a demand function $h$, we define $P_{h,j} := \{\cakeprt\in\Delta^{m-1} ~|~ j\in h(\cakeprt)\}$ to be the set of partitions in which piece~$j$ is preferred according to $h$.
If $h$ satisfies the closed preference sets assumption, then $P_{h,j}$ is a closed subset of $\simplex$ for every $j$.

The proof idea is to convert any chore-allocation instance to a cake-allocation instance. 
Instead of giving each point of the chore $[0,1]$ to a single group, 
we give \emph{exemptions} from the chore to $m-1$ groups. 
We show that, if the players' preferences in the chore instance satisfy the lazy players and closed preference sets assumptions, 
then their preferences in the exemptions instance satisfy the hungry players and closed preference sets assumptions, and therefore the existence result for cake cutting is applicable. The details are given below.

\begin{proof}[Proof of Theorem \ref{thm:chores}]
For each $m$-partition $\cakeprt\in\Delta^{m-1}$, define a vector $\choreprt\in\mathbb{R}^m$ by
\begin{align*}
\choreprt := \mathbf{1}_m - (m-1)\cdot \cakeprt,
\end{align*}
where $\mathbf{1}_m$ is a vector of $m$ ones, i.e., $\widehat{x}_j := 1 - (m-1)\cdot x_j$ for all $j\in[m]$.
Note that $\sum_{j=1}^m \widehat{x}_j=m-(m-1)=1$, and $\choreprt\in \simplex$ if and only if $x_j\leq \frac{1}{m-1}$ for all $j\in[m]$.
Let $\smallsimplex:=\{\cakeprt\in\simplex ~|~ 
\forall j\in[m]: x_{j}\leq \frac{1}{m-1}
\}$.
Intuitively, for any point $\cakeprt\in\smallsimplex$,
the vector $(m-1)\cdot \cakeprt$ represents an allocation of exemptions, and $\choreprt$ represents the corresponding allocation of chores.

We are given $n$ lazy players, where each lazy player $i\in[n]$ has a demand function $\ell_i: \Delta^{m-1}\to 2^{[m]}$.
From these players, we construct $n$ \emph{hungry players}, with demand functions $h_i: \Delta^{m-1}\to 2^{[m]}$ defined as follows:
\begin{align*}
h_i(\cakeprt) := 
\begin{cases}
\ell_i(\choreprt) & \text{if } \cakeprt\in \smallsimplex;
\\
\left\{j \,\middle|\, x_j \geq  \frac{1}{m-1}\right\} & \text{otherwise}.
\end{cases}
\end{align*}
Note that $h_i(\cakeprt)$ is nonempty in both cases.
In the first case, $(m-1)\cdot \cakeprt$ corresponds to a valid allocation of exemptions; in the second case, $(m-1)\cdot \cakeprt$ corresponds to an ``over-allocation'' of exemptions, that is, an allocation in which some exemptions (those with $x_j>\frac{1}{m-1}$) span more than the entire cake. 

In order to apply the cake-cutting theorem from \cite{SegalhaleviSu21}, we now show that $h_i$ satisfies the hungry players and closed preference sets assumptions.

\underline{Hungry players}.
Consider an $m$-partition $\cakeprt$ where $x_j=0$ for some $j\in[m]$. 
We have to show that $j\not\in h_i(\cakeprt)$.
If $\cakeprt\not\in\smallsimplex$,
then by definition $h_i(\cakeprt)$ contains only pieces of length at least $\frac{1}{m-1}$, so $j\not\in h_i(\cakeprt)$.
Suppose now that $\cakeprt\in\smallsimplex$.
Since $\sum_{j=1}^m x_j = 1$, the definition of $\smallsimplex$ implies that $x_{j'} = \frac{1}{m-1}$ for all $j'\ne j$.
Hence, $\widehat{x}_j = 1$, and $\widehat{x}_{j'} = 0$ for all $j'\neq j$.
Since $\ell_i$ represents a lazy player, 
who never prefers a positive piece if an empty piece is available,
$j\not\in \ell_i(\choreprt)= h_i(\cakeprt)$.

\underline{Closed preference sets}. Fix some $i\in[n]$ and $j\in[m]$. The definition of $h_i$ implies that $P_{h_i,j} = Q_{h_i, j} \cup R_{h_i, j}$, where
\begin{align*}
Q_{h_i, j} &:= \{\cakeprt\in\smallsimplex ~|~ 
j\in \ell_i(\choreprt)\};
\\
R_{h_i, j} &:= \left\{\cakeprt\not\in\delta^{m-1} \,\middle|\,
x_j\geq \tfrac{1}{m-1}\right\}.
\end{align*}

The set $Q_{h_i, j}$ is closed since it is a linear transformation of $P_{\ell_i,j}$, which is closed by the closed preference sets assumption on the original lazy-players instance. 
Specifically, recall that 
$
P_{\ell_i,j} = \{\cakeprt\in\simplex ~|~ 
j\in \ell_i(\cakeprt)\}$.
If we take the reflection of $P_{\ell_i,j}$ over the origin, translate it by $\mathbf{1}_m$, and dilate it by $\frac{1}{m-1}$, we get another closed set:
\begin{align*}
&
\left\{ \tfrac{1}{m-1}\cdot(\mathbf{1}_m - \cakeprt) \,\middle|\,
\cakeprt \in P_{\ell_i,j}
\right\}
\\
&=
\{ \mathbf{y} \in\mathbb{R}^m ~~|~~
\mathbf{1}_m  - (m-1)\cdot \mathbf{y} \in P_{\ell_i,j}
\}
\\
&=
\{ \mathbf{y} \in\mathbb{R}^m ~~|~~
\widehat{\mathbf{y}} \in P_{\ell_i,j}
\}
\\
&=
\{ \mathbf{y} \in\mathbb{R}^m ~~|~~
\widehat{\mathbf{y}} \in \simplex
~,~
j\in \ell_i(\widehat{\mathbf{y}})
\}
\\
&=
\{ \mathbf{y} \in\mathbb{R}^m ~~|~~
\mathbf{y}\in \smallsimplex
~,~
j\in \ell_i(\widehat{\mathbf{y}})
\}
&&\text{(since $\widehat{\mathbf{y}} \in \simplex$  iff $\mathbf{y}\in\smallsimplex$)}
\\
&= Q_{h_i, j}.
\end{align*}
Hence, $Q_{h_i, j}$ is closed.

The set $R_{h_i, j}$ by itself is not closed, but its boundary is contained in $Q_{h_i,j}$. Specifically, by Lemma~\ref{lem:strong-lazy}, 
for any $\cakeprt\in\smallsimplex$ (i.e., $\choreprt\in\simplex$),
$\ell_i(\choreprt)$ contains all $j$ for which $\widehat{x}_j=0$, so it contains all $j$ for which $x_j=\frac{1}{m-1}$. 
Therefore,
\begin{align*}
&
Q_{h_i, j} \supseteq 
\left\{
\cakeprt\in\smallsimplex
\,\middle|\,
x_j = \tfrac{1}{m-1}
\right\}
=
\left\{
\cakeprt\in\smallsimplex
\,\middle|\,
x_j \geq \tfrac{1}{m-1}
\right\}
\\
\implies&
Q_{h_i, j}\cup R_{h_i, j} = 
Q_{h_i, j} \cup 
\left\{\cakeprt\in\simplex \,\middle|\, 
x_j \geq \tfrac{1}{m-1}
\right\},
\end{align*}
so $P_{h_i,j}$ is the union of two closed sets and is consequently closed itself.

We are now ready to complete our proof.
By Theorem 2 in  \cite{SegalhaleviSu21},
there exists a partition $\mathbf{y}\in\simplex$, along with a division of the hungry players (represented by the demand functions $h_i$) into groups with group~$j$ containing $k_j$ players, such that the partition is envy-free for the hungry players.
If $\mathbf{y}\not\in\smallsimplex$, then all hungry players prefer only pieces with length at least $\frac{1}{m-1}$, and there are at most $m-1$ such pieces, so at least one group is assigned a non-preferred piece, a contradiction with envy-freeness.
Hence, $\mathbf{y}\in\smallsimplex$. 
By construction, the partition 
$\mathbf{\widehat{y}}$ along with the same division of the lazy players into groups is envy-free for the lazy players.
\end{proof}

\section{Mixed Cake.}
\label{sec:mixed}

Cake cutting assumes that all players are \emph{hungry} and never prefer an empty piece, whereas chore division assumes that all of them are \emph{lazy} and never prefer a nonempty piece over an empty one.
What if neither of these assumptions are universally true---perhaps some employees would rather stay home than spending certain periods at the cottage, and some guards enjoy making sure that the public building is secured?
The division of a so-called ``mixed cake'' among individual players has been studied by several authors in the last few years \cite{avvakumov2021envy,BogomolnaiaMoSa17,meunier2019envy,Segalhalevi18}.

The mixed cake model makes neither the hungry players nor the lazy players assumptions, while keeping the closed preference sets assumption.
Without additional assumptions, 
an envy-free allocation might not exist even for two singleton groups---for example, when both players always prefer only the right piece regardless of the cut point. While this preference satisfies the closed preference sets assumption, it implies that the players prefer the whole cake to an empty cake (when the cut point is $0$) and also prefer an empty cake to the whole cake (when the cut point is $1$)! To avoid this unrealistic situation, we add the requirement that that if two partitions $(x_1,\dots,x_m)$ and $(x_1',\dots,x_m')$ give rise to the same split of the cake, then each player's preference should be consistent across the partitions.
For instance, if in the partition $(0,0.2,0.8)$ a player prefers only the piece of length $0.8$, then the player should prefer only the piece of length $0.8$ in the partitions $(0.2, 0, 0.8)$ and  $(0.2, 0.8, 0)$ too.
However, even with this requirement, the existence of an envy-free assignment still cannot be guaranteed.
Indeed, consider five players whose preferences are represented by additive cardinal utility functions shown in Figure~\ref{fig:mixed-cake}.
The utilities are uniformly distributed within each labeled piece.
Given any partition of the mixed cake into two pieces, each player prefers the piece that yields a higher utility to the player (or both pieces in case of a tie).
It is clear that the closed preference sets assumption is satisfied for all players, and an unrealistic situation as above no longer occurs.
On the other hand, one can check that in any $2$-partition, each piece is preferred by no more than three players, meaning that an envy-free assignment for two groups of sizes four and one does not exist.

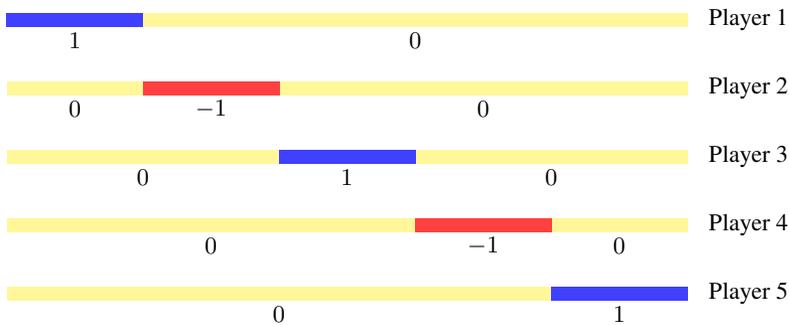
\begin{figure}[!ht]
\centering
\begin{tikzpicture}[scale=0.9]
\node at (11.9,1.1) {\small Player 5};
\path [fill=yellow!50] (1,1) rectangle (11,1.2);
\path [fill=blue!75] (9,1) rectangle (11,1.2);
\node at (10,0.8) {\small $1$};
\node at (5,0.8) {\small $0$};
\node at (11.9,2.1) {\small Player 4};
\path [fill=yellow!50] (1,2) rectangle (11,2.2);
\path [fill=red!75] (7,2) rectangle (9,2.2);
\node at (8,1.8) {\small $-1$};
\node at (4,1.8) {\small $0$};
\node at (10,1.8) {\small $0$};
\node at (11.9,3.1) {\small Player 3};
\path [fill=yellow!50] (1,3) rectangle (11,3.2);
\path [fill=blue!75] (5,3) rectangle (7,3.2);
\node at (6,2.8) {\small $1$};
\node at (3,2.8) {\small $0$};
\node at (9,2.8) {\small $0$};
\node at (11.9,4.1) {\small Player 2};
\path [fill=yellow!50] (1,4) rectangle (11,4.2);
\path [fill=red!75] (3,4) rectangle (5,4.2);
\node at (4,3.8) {\small $-1$};
\node at (2,3.8) {\small $0$};
\node at (8,3.8) {\small $0$};
\node at (11.9,5.1) {\small Player 1};
\path [fill=yellow!50] (1,5) rectangle (11,5.2);
\path [fill=blue!75] (1,5) rectangle (3,5.2);
\node at (2,4.8) {\small $1$};
\node at (7,4.8) {\small $0$};
\end{tikzpicture}
\caption{An example showing that Theorem~\ref{thm:chores} does not extend to a mixed cake.}
\label{fig:mixed-cake}
\end{figure}

This counterexample implies that in order to recover the guaranteed existence of an envy-free assignment, it may be necessary to allocate more than one contiguous piece of the mixed cake to each group.
The preferences in our model thus far are defined only for partitions into $m$ contiguous pieces, where $m$ is the number of groups.
We therefore extend them to partitions into $m$ \emph{collections} of pieces, where each collection consists of a finite number of contiguous pieces.
As before, for any such partition, each player prefers one or more collections, and the corresponding preference sets are closed.
For several division problems, it is desirable to allocate a small number of pieces to each group.
Unfortunately, the example in Figure~\ref{fig:mixed-cake} can be generalized to show that a number of cuts that is linear in the number of players may be inevitable even in the case of two groups.

\begin{theorem}
\label{thm:mixed-cake-negative}
Let $n\ge 4$ be any integer, $m=2$, and $k_1,k_2$ be positive integers with $k_1+k_2 = n$ and $\min\{k_1,k_2\} = 1$.
There exist $n$ players with preferences represented by additive cardinal utility functions such that any envy-free assignment of the mixed cake with group $j$ containing $k_j$ players requires at least $n-3$ cuts.
\end{theorem}

\begin{proof}
Assume without loss of generality that $k_1 = n-1$ and $k_2 = 1$.
Consider $n$~players whose preferences are represented by additive cardinal utility functions.
For odd $i\in[n]$, player~$i$ has utility $1$ for the interval $[\frac{i-1}{n}, \frac{i}{n}]$, and utility $0$ for other intervals; the utilities are uniformly distributed within each interval.
Similarly, for even $i\in[n]$, player~$i$ has utility $-1$ for the interval $[\frac{i-1}{n}, \frac{i}{n}]$, and utility $0$ for other intervals.
The case $n=5$ is illustrated in Figure~\ref{fig:mixed-cake}.

Consider an envy-free allocation, and 
let $C$ be the piece (union of intervals) preferred by $n-1$ players.
Call a pair of players with adjacent indices $i$ and $i+1$ \emph{conforming} if both of them prefer $C$.
Note that, of the $n-1$ adjacent pairs of players, at least $n-3$ pairs are conforming.
For each conforming pair $\{i,i+1\}$, if $i$ has utility~$1$ for the interval $[\frac{i-1}{n}, \frac{i}{n}]$, there must be a cut in the interval $(\frac{i-1}{n},\frac{i+1}{n})$ such that the piece to the left of the cut belongs to $C$, while the piece to the right does not.
Analogously, if player~$i$ has utility $-1$ for the interval $[\frac{i-1}{n}, \frac{i}{n}]$, there must be a cut in the interval $(\frac{i-1}{n},\frac{i+1}{n})$ such that the piece to the left of the cut does not belong to $C$, while the piece to the right does.
Since no two cuts corresponding to different conforming pairs can coincide, the number of cuts is at least $n-3$.
\end{proof}

On the positive side, by using a classical result on ``consensus halving,'' one can see that the $n-3$ cuts required by Theorem~\ref{thm:mixed-cake-negative} is almost as bad as it gets, even for an arbitrary distribution of the players into the two groups.
Formally, for any positive integers $k_1,k_2$ such that $k_1+k_2 = n$, for $n$ players with preferences represented by continuous cardinal utility functions (not necessarily additive or monotone), there exists an envy-free assignment of the mixed cake with group $j$ containing $k_j$~players that uses at most $n$ cuts.
Indeed, a result of Simmons and Su~\cite{SimmonsSu03} states that for $n$~players with such utility functions, there exists a partition of the mixed cake into two parts using at most $n$ cuts such that every player has equal utility for both parts.
This partition allows us to divide the players into two groups arbitrarily so that group~$j$ contains $k_j$ players.\footnote{In fact, $n-1$ cuts suffice, since we can compute a consensus halving for $n-1$ players and let the $n$th player pick a preferred piece for his or her group.}

As our final remark, we note that the vast majority of the cake-cutting literature assumes that players' preferences are captured by additive utility functions---indeed, additivity is portrayed as an inherent part of the model in surveys on the subject~\cite{Procaccia13,Procaccia16}. 
However, in many applications of cake cutting the preferences are neither additive nor monotone, and as our work demonstrates, allowing general preferences raises several new questions with intriguing answers.
Investigating other questions in the more general preference model is therefore an appealing direction for future work.

\begin{acknowledgment}{Acknowledgment.}
This work was partially supported by the Israel Science
Foundation under grant number 712/20, by the Singapore Ministry of Education under grant number MOE-T2EP20221-0001, and by an NUS
Start-up Grant.
The authors wish to thank the editor and the anonymous reviewers for several constructive comments.
\end{acknowledgment}

\begin{affil}
Department of Computer Science, Ariel University, Israel\\
erelsgl@gmail.com
\end{affil}

\begin{affil}
School of Computing, National University of Singapore, Singapore\\
warut@comp.nus.edu.sg
\end{affil}


\begin{thebibliography}{99}

\bibitem{avvakumov2021envy}
Avvakumov, S., Karasev, R. (2021).
Envy-free division using mapping degree.
\textit{Mathematika}. 67(1): 36--53.

\bibitem{BogomolnaiaMoSa17} Bogomolnaia, A., Moulin, H., Sandomirskiy, F., Yanovskaya, E. (2017). Competitive division of a mixed manna. \textit{Econometrica}. 85(6): 1847--1871.

\bibitem{BramsTa96} Brams, S. J., Taylor, A. D. (1996). \textit{Fair Division: From Cake-Cutting to Dispute Resolution}. Cambridge, UK: Cambridge University Press.

\bibitem{GhodsiLaMo18}
Ghodsi, M., Latifian, M., Mohammadi, A., Moradian, S., Seddighin, M.
  (2018).
\newblock Rent division among groups.
\newblock In: Kim, D., Uma, R. N., Zelikovsky, A., eds. {\em Proceedings of the 12th International Conference on Combinatorial Optimization and
  Applications}.
  Cham: Springer, pp. 577--591.

\bibitem{KyropoulouSuVo20} Kyropoulou, M., Suksompong, W., Voudouris, A. A. (2020). Almost envy-freeness in group resource allocation. 
\textit{Theor. Comput. Sci}. 841: 110--123.

\bibitem{ManurangsiSu17} Manurangsi, P., Suksompong, W. (2017). Asymptotic existence of fair divisions for groups. \textit{Math. Soc. Sci}. 89: 100--108.

\bibitem{ManurangsiSu22} Manurangsi, P., Suksompong, W. (2022). Almost envy-freeness for groups: Improved bounds via discrepancy theory. \textit{Theor. Comput. Sci}. 930: 179--195.

\bibitem{meunier2019envy}
Meunier, F., Zerbib, S. (2019).
Envy-free cake division without assuming the
  players prefer nonempty pieces.
\textit{Israel J. Math}. 234(2): 907--925.

\bibitem{Procaccia13} Procaccia, A. D. (2013). Cake cutting: Not just child's play. \textit{Commun. ACM}. 56(7): 78--87.

\bibitem{Procaccia16} Procaccia, A. D. (2016). Cake cutting algorithms. In: Brandt, F., Conitzer, V., Endriss, U., Lang, J., Procaccia, A. D., eds. \emph{Handbook of Computational Social Choice}. Cambridge, UK: Cambridge University Press, pp. 311--329.

\bibitem{RobertsonWe98} Robertson, J., Webb, W. (1998). \textit{Cake-Cutting Algorithms: Be Fair if You Can}. Boca Raton, FL: Peters/CRC Press.

\bibitem{Segalhalevi18} Segal-Halevi, E. (2018). Fairly dividing a cake after some parts were burnt in the oven. In: Andr\'{e}, E., Koenig, S., Dastani, M., Sukthankar, G., eds. \textit{Proceedings of the 17th International Conference on Autonomous Agents and MultiAgent Systems}. Richland, SC: International Foundation for Autonomous Agents and Multiagent Systems, pp. 1276--1284.

\bibitem{SegalhaleviNi19} Segal-Halevi, E., Nitzan, S. (2019). Fair cake-cutting among families. \textit{Soc. Choice Welf.} 53(4): 709--740.

\bibitem{SegalhaleviSu19} Segal-Halevi, E., Suksompong W. (2019). Democratic fair allocation of indivisible goods. \textit{Artif. Intell.} 277: 103167.

\bibitem{SegalhaleviSu21} Segal-Halevi, E., Suksompong W. (2021). How to cut a cake fairly: a generalization to groups. \textit{Amer. Math. Monthly}. 128(1): 79--83.

\bibitem{SimmonsSu03} Simmons, F. W., Su, F. E. (2003). Consensus-halving via theorems of Borsuk-Ulam and Tucker. \textit{Math. Soc. Sci}. 45(1): 15--25.

\bibitem{Stromquist80} Stromquist, W. (1980). How to cut a cake fairly. \textit{Amer. Math. Monthly}. 87(8): 640--644. Addendum at \textit{Amer. Math. Monthly}. 88(8): 613--614.

\bibitem{Stromquist08} Stromquist, W. (2008). Envy-free cake divisions cannot be found by finite protocols. \textit{Electron. J. Comb}. 15: \#R11.

\bibitem{Su99} Su, F. E. (1999). Rental harmony: Sperner's lemma in fair division. \textit{Amer. Math. Monthly}. 106(10): 936--942.

\bibitem{Suksompong18} Suksompong, W. (2018). Approximate maximin shares for groups of agents. \textit{Math. Soc. Sci}. 92: 40--47.

\bibitem{Suksompong18-dissertation} Suksompong, W. (2018). Resource allocation and decision making for groups. PhD dissertation. Stanford University, Stanford, USA.

\bibitem{Woodall80} Woodall, D. R. (1980). Dividing a cake fairly. \textit{J. Math. Anal. Appl}. 78(1): 233--247.



\end{thebibliography}
\end{document}